%% file: paper__TCOM_v3.tex
\documentclass[journal,10pt]{IEEEtran}
\usepackage{cite}
\usepackage{graphicx,color,epsfig,rotating}
\usepackage{amsfonts,amsmath,amssymb,bbm}
\usepackage{setspace}
\usepackage{algorithm}
\usepackage[algo2e]{algorithm2e} 
\usepackage{subcaption}
\usepackage{mdwtab}
\usepackage{placeins} 
\usepackage{psfrag, graphicx}
\usepackage[latin1]{inputenc}
\usepackage{multirow}
\usepackage{stfloats}
\usepackage{tabularx} 
\usepackage{booktabs} 
\usepackage{url}
\usepackage{bm}
\usepackage{pstricks}
 \allowdisplaybreaks

\input{macros}

\newcommand{\xz}[1]{{\color{red} [XZ COMMENT: #1]}}

\title{Information-Theoretic Decentralized Secure Aggregation with User Dropouts}
\input{author_journal}

\begin{document}
\maketitle
\IEEEpeerreviewmaketitle

\begin{abstract}
This paper investigates the fundamental limits of information-theoretic decentralized secure aggregation (DSA) with user dropouts. We consider a fully decentralized network where $K$ users communicate over broadcast channels without a trusted aggregation server. 
Each user holds a private input and aims to recover the sum of the surviving users' inputs (users may drop) while ensuring that no additional information about individual inputs is revealed  to that user, 
even if it can collude with other users.
A two-round communication protocol is considered, where we assume at least $U$ users survive  and each user can collude with  at most $T$ other users.
For this setting, the optimal communication rate region is fully characterized:  we show that DSA  is infeasible if $U\le T+1$; otherwise, the optimal rate region is given by $R_1\geq 1$ and $R_2\geq \frac{1}{U-T-1}$, where $R_1$ and $R_2$ denote the first- and second-round communication rates, respectively. 
The proposed  aggregation scheme is based on correlated secret keys constructed from $(T+1)$-private maximum distance separable (MDS) matrices, which simultaneously provide robustness against user dropouts and security against collusion. 
We also derive tight converse bounds that establish  the optimality of the proposed scheme.
Our result shows that  the optimal second-round communication rate depends only on the effective redundancy level $U-T-1$ regardless the total number of users.
\end{abstract}


\section{Introduction}
The rapid growth of distributed learning and collaborative computation has significantly accelerated the development of federated learning (FL) and related distributed optimization frameworks~\cite{konecny2016federated,pmlr-v54-mcmahan17a,KairouzFL,9084352,yang2018applied}. In these systems, users collaboratively compute global statistics or train machine learning models while keeping their local datasets private. A fundamental primitive enabling such privacy-preserving computation is \emph{secure aggregation} (SA)~\cite{bonawitz2016practical,bonawitz2017practical}, where the sum of users' inputs (\eg, local models in FL) is computed by an \agg  server while  the server is prevented from  inferring \indiv inputs 
beyond the desired sum.
The information-theoretic formulation of secure aggregation was first studied by Zhao and Sun~\cite{zhao2023secure}. Since then, substantial efforts have been devoted to improving robustness, communication efficiency, and privacy guarantees~\cite{liu2022efficient,zhao2022mds,li2025weakly,Li_Zhang_WeaklyHSA,zhang_mshsa_online2025,Zhang_Wan_HSA,Li_Zhang_MSSA}. 
Most existing works focus on centralized architectures in which a server coordinates the aggregation process.

Recently, decentralized secure aggregation (DSA)~\cite{Zhang_Li_Wan_DSA,Li_Zhang_GroupwiseDSA,Li_Zhang_WeaklyDSA} has attracted increasing attention, motivated by scenarios where relying on a trusted server is undesirable due to scalability, privacy, or communication constraints. In decentralized systems, users communicate directly with one another and collaboratively recover the desired aggregation. Compared with the centralized setting, DSA introduces new challenges because every user simultaneously acts as both a transmitter and a decoder.
A major practical challenge in DSA is robustness against user \emph{dropouts}~\cite{9834981,so2022lightsecagg,jahani2022swiftagg,li_Zhang_HSAdropout,zhang2025secure}, where users may fail  due to unreliable communication links, device failures, or asynchronous participation. Since the set of surviving users is generally unknown in advance, the aggregation protocol must guarantee reliable input sum recovery under arbitrary dropout patterns. Another important challenge is security
under \emph{collusion}~\cite{Li_Zhang_GroupwiseDSA,xu2025HSA,Li_Zhang_WeaklyDSA}, where multiple users may collude to infer information about non-colluding users' private inputs.

Motivated by the above considerations, this paper studies the fundamental communication limits of information-theoretic DSA with user dropouts. We consider a fully connected network consisting of $K$ users connected through broadcast channels. Communication takes place in \tit{two rounds}. In the first round, each user broadcasts a message generated from its input and secret keys.
Due to user dropouts, only a subset $\mathcal{U}_1 \subseteq \{1, \cdots,K\}$ of users successfully survives the first round. In the second round, additional dropouts may also occur, resulting in a surviving set $\mathcal{U}_2\subseteq \mathcal{U}_1$. After the two rounds, 
every user in $\mathcal{U}_2$ aims to recover the aggregate
$\sum_{k\in\mathcal{U}_1} W_k$, where $W_k$ denotes the private input of user $k$.

The goal of this work is to characterize the optimal communication rates of DSA under user dropout and collusion constraints. We establish a  feasibility threshold showing that DSA is feasible if and only if
$U>T+1$.
In this regime, we characterize the optimal communication rate region which consists of all  \achvb rates tuples of the two rounds.
In particular, we show that the optimal first-round communication rate is
$
R_1^*=1
$. 
while the optimal second-round communication rate is
$
R_2^*=\frac{1}{U-T-1}
$.
Our achievable scheme is based on vector linear coding and correlated keys generated using $(T+1)$-private MDS matrices. The converse proofs rely on entropy inequalities and Shannon-theoretic secrecy arguments~\cite{Shannon1949,Sun_Jafar_SPIR,Jia_Sun_Jafar}. Compared with existing secure aggregation protocols~\cite{user_held,aggregation,aggregation_log,aggregation_turbo,aggregation_fast,bonawitz2019federated,Choi_Sohn_Han_Moon},
the proposed framework provides information-theoretic security guarantees while explicitly capturing decentralized communication, arbitrary user dropout, and collusion resistance within a unified framework.

\subsection{Summary of Contributions}

The main contributions of this paper are:

\begin{itemize}

\item We study a two-round DSA problem with information-theoretic security guarantee under user dropout and collusion. Unlike conventional centralized secure aggregation systems that rely on a trusted server to coordinate decoding, the considered decentralized model requires each surviving user to simultaneously act as both a transmitter and a decoder through peer-to-peer broadcast communications. This fundamentally changes the structure of the problem, since robustness against dropout and collusion must now be guaranteed for all surviving users without  the presence of a centralized \agg server.

\item We characterize the feasibility condition of DSA with dropout. Specifically, we prove that secure aggregation is feasible if and only if
$U>T+1$. 
The result reveals a sharp threshold determined by the interplay between the minimum number of surviving users and the collusion level. In the feasible regime, we completely characterize the optimal communication rate region by deriving matching achievability and converse bounds.

\item We develop an optimal achievable scheme based on vector linear coding and correlated keys generated from $(T+1)$-private MDS matrices. The proposed construction enables every surviving user to recover the desired aggregation despite arbitrary dropout patterns, while simultaneously preserving information-theoretic privacy against colluding users. The converse proofs combine entropy inequalities and Shannon-theoretic secrecy arguments to establish the optimality of the proposed communication rates.

\item Our results show that the optimal communication efficiency depends only on the effective redundancy level
$
U-T-1
$,
and is independent of the total number of users $K$. This provides a new insight into DSA: scalability does not  affect the optimal communication cost, whereas the balance between dropout tolerance and collusion resistance completely determines the achievable efficiency.

\end{itemize}

\section{Problem Statement}
\label{sec: problem description}

We consider a DSA problem with $K \geq 3$ users, where each user is connected to all others via an error-free broadcast channel, as illustrated in Fig.~\ref{fig:model}.
\begin{figure}[ht]
    \centering
    \includegraphics[width=0.48\textwidth]{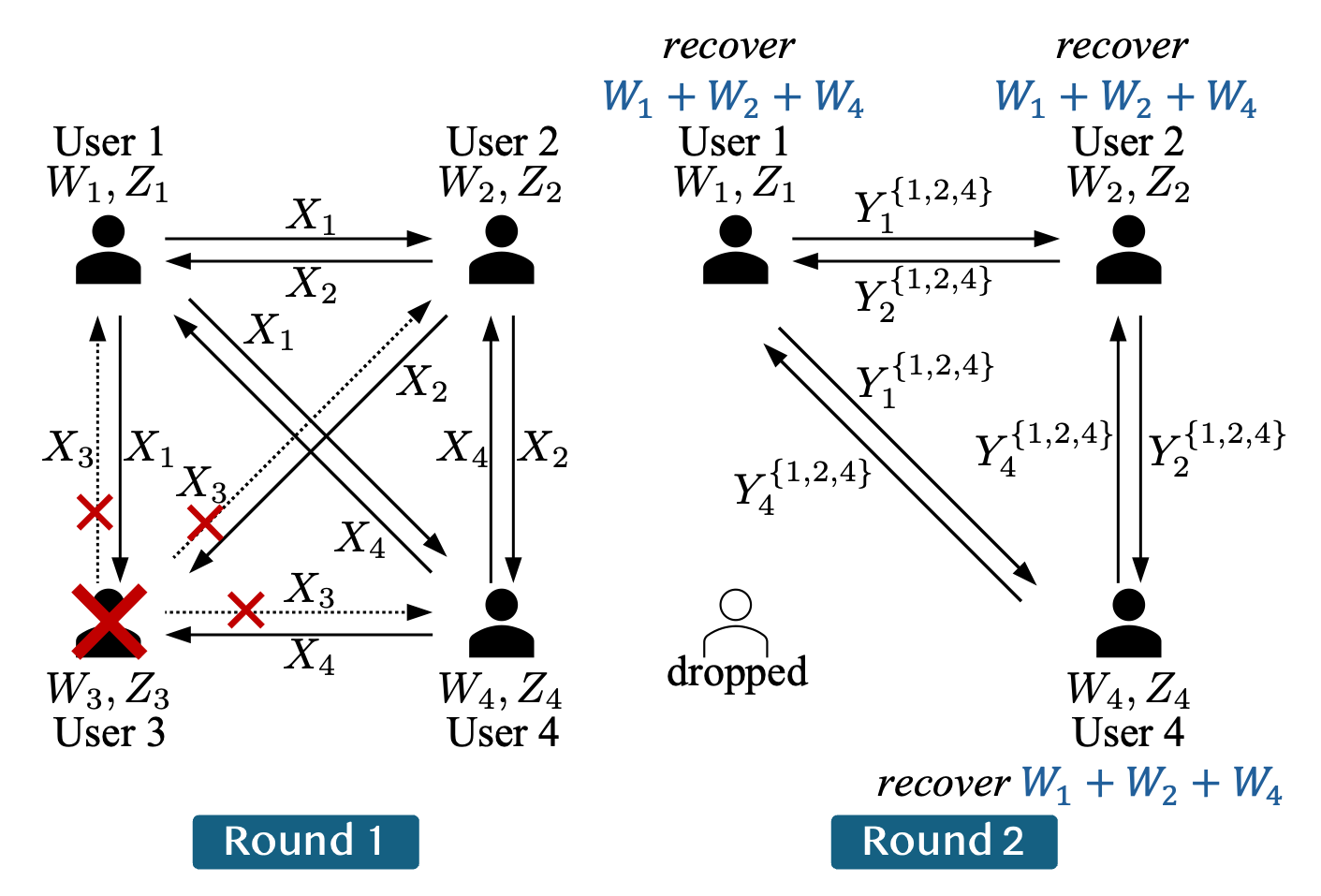}
    \caption{DSA with $K=4$, $U=3$, and $T=0$. In the first round, each user $k$ broadcasts $X_k$ to all other users. In this example, User~3 drops out over first round, and hence the set of surviving user is $\mathcal{U}_1=\{1,2,4\}$. In the second round, each surviving user $k\in\mathcal{U}_1$ broadcasts $Y_k^{\mathcal{U}_1}$. From the received messages and its own input and secret key, each surviving user must recover $W_1+W_2+W_4$ while learn nothing more.}
    \label{fig:model}
\end{figure}
Each user $k \in [K] \triangleq \{1,2,\ldots,K\}$ holds an input vector $W_k$ and a key variable $Z_k$. 
The inputs $\{W_k\}_{k\in[K]}$ are mutually independent, and each $W_k$ is an $L \times 1$ vector whose entries are i.i.d. uniformly distributed over $\mathbb{F}_q$. Furthermore, $\{W_k\}_{k\in[K]}$ is independent of $\{Z_k\}_{k\in[K]}$, i.e.,
\begin{align}
   &H\left(\{W_k\}_{k\in[K]}, \{Z_k\}_{k\in[K]}\right)\notag\\
   =& \sum_{k\in[K]} H(W_k) + H\!\left(\{Z_k\}_{k\in[K]}\right), \label{ind} \\
   &H(W_k) = L \quad (\text{in $q$-ary units}), \quad \forall k \in [K]. \label{h2}
\end{align}

To enable the computation of the desired sum, we consider the following two-round \comm protocol.

\textbf{First Round:}
Each user $k \in [K]$ broadcasts a message $X_k$, consisting of $L_X$ symbols over $\mathbb{F}_q$, which is a deterministic function of its local input and key:
\begin{align}
H(X_k \mid W_k, Z_k) = 0, \quad \forall k \in [K]. \label{messageX}
\end{align}
Due to possible user dropouts\footnote{A user is said to drop out if it fails to participate in the broadcast phase, i.e., it neither transmits nor receives any message.},
the set of users whose transmissions are successfully received is denoted by $\mathcal{U}_1 \subseteq [K]$, where $|\mathcal{U}_1| \geq U$\footnote{$U$ is a design parameter specifying the minimum number of users required for successful decoding.} and $1 \leq U \leq K-1$. 
The set $\mathcal{U}_1$ is arbitrary subject to the cardinality constraint.
Each user $u \in [K]$ observes the messages $\{X_k\}_{k\in \mathcal{U}_1\setminus\{u\}}$.

\textbf{Second Round:}
Each user determines its participation status based on its first-round observations. 
A user that fails to receive the required messages declares itself as dropped out and remains silent; otherwise, it is active and can infer the surviving set $\mathcal{U}_1$. 
Only users in $\mathcal{U}_1$ participate in the second round.
Upon identifying $\mathcal{U}_1$, each user $k \in \mathcal{U}_1$ broadcasts a message $Y^{\mathcal{U}_1}_k$ of length $L_Y$ over $\mathbb{F}_q$, satisfying
\begin{align}
H\Big(Y^{\mathcal{U}_1}_k \Big| W_k, Z_k\Big) = 0, 
\quad \forall k \in \mathcal{U}_1,\ \forall \mathcal{U}_1 \subseteq [K]: |\mathcal{U}_1| \geq U. 
\label{messageY}
\end{align}
Further dropouts may occur in the second round, and the set of surviving users is denoted by $\mathcal{U}_2$, where $\mathcal{U}_2 \subseteq \mathcal{U}_1$ and $|\mathcal{U}_2| \geq U$. 
Each user $u \in \mathcal{U}_2$ receives the messages $\{Y^{\mathcal{U}_1}_k\}_{k \in \mathcal{U}_2 \setminus \{u\}}$.

\textbf{Correctness:}
From the messages received in both rounds from surviving users, together with its own information, each user $u \in \mathcal{U}_2$ must be able to decode the desired sum $\sum_{k \in \mathcal{U}_1} W_k$ with zero error\footnote{The results also hold under vanishing error and leakage, i.e., when 0 is replaced by $o(L)$ in (\ref{correctness}) and (\ref{security}).}.
Formally, for all $\mathcal{U}_2 \subseteq \mathcal{U}_1 \subseteq [K]$ with $|\mathcal{U}_2| \geq U$,
\begin{align}
\mbox{[Correctness]}~~
&H\Bigg(\sum_{k \in \mathcal{U}_1} W_k \,\Bigg|\, 
\{X_k\}_{k \in \mathcal{U}_1 \setminus \{u\}}, 
\{Y^{\mathcal{U}_1}_k\}_{k \in \mathcal{U}_2 \setminus \{u\}}, \notag\\
&\qquad W_u, Z_u \Bigg) = 0, \quad \forall u \in \mathcal{U}_2.
\label{correctness}
\end{align}

\textbf{Security:}
We require that security holds for any user $u \in [K]$\footnote{The adversary can be any user, regardless of whether it belongs to the surviving set $\mathcal{U}_1$.}, even if it observes all transmitted messages and may collude with any set of at most $T$ users, where $1 \leq T \leq K-2$.  
Specifically, user $u$ should not learn any additional information about $\{W_k\}_{k \in [K]}$ beyond the sum $\sum_{k \in \mathcal{U}_1} W_k$ and the information available from colluding users. Formally, for any $\mathcal{U}_1, \mathcal{T} \subseteq [K]$ with $|\mathcal{U}_1| \geq U$ and $|\mathcal{T}| \leq T$,
\begin{align}
\mbox{[Security]}~
&I\Bigg(\{W_k\}_{k \in [K]}; 
\{X_k\}_{k \in [K]\setminus \{u\}}, 
\{Y^{\mathcal{U}_1}_k\}_{k \in \mathcal{U}_1 \setminus \{u\}} \Bigg| \notag\\
& \sum_{k \in \mathcal{U}_1} W_k, 
\{W_k, Z_k\}_{k \in \mathcal{T} \cup \{u\}} \Bigg) = 0, 
 \forall u \in [K].
\label{security}
\end{align}

\begin{remark}[Delayed Message Availability at User $u$]
In the security constraint, user $u$ is assumed to have access to all first-round messages, including those from users that subsequently drop out, as well as all second-round messages from surviving users. This models a delayed message availability scenario in which messages from dropped users may eventually become available, resulting in a strictly stronger security requirement.
\end{remark}

\textbf{Rates:}
The communication rates are defined as
\begin{align}
R_1 \triangleq \frac{L_X}{L}, \quad 
R_2 \triangleq \frac{L_Y}{L},
\label{rate}
\end{align}
where $R_1$ and $R_2$ denote the first-round and second-round message rates, respectively.
A rate pair $(R_1, R_2)$ is said to be achievable if there exists a secure aggregation scheme, i.e., a design of the secret keys $\{Z_k\}_{k \in [K]}$ and encoding functions for $\{X_k\}_{k \in [K]}$ and $\{Y^{\mathcal{U}_1}_k\}_{k \in [K]}$, that achieves the rates  $R_1$ and $R_2$ simultaneously while satisfying the correctness and security constraints in (\ref{correctness}) and (\ref{security}).
The  optimal rate region $\mathcal{R}^*$ is defined as
the closure of all achievable rate pairs.

\section{Main Result}\label{sec:result}

Theorem~\ref{thm:main} characterizes the optimal communication rate region for the two-round DSA problem.

\begin{theorem}\label{thm:main}
Consider the DSA problem with $K \geq 3$ users, where at least $U$ users remain active and at most $T$ users may collude, with $1 \leq U \leq K-1$ and $0 \leq T \leq K-3$. The optimal rate region is given by
\begin{align}
\Rc^*=
\begin{cases}
\emptyset & \trm{if}~U \leq T+1, \\
\big\{ (R_1, R_2) :  R_1 \geq 1, R_2 \geq \frac{1}{U-T-1} \big\} & \trm{if}~ U > T+1.
\end{cases}
\end{align}
\end{theorem}

From Theorem~\ref{thm:main} and its proof (see Section IV for achievability and Section V for converse), we provide the following interpretations:

\tit{1) Infeasibility:} \Thm \ref{thm:main} reveals a  feasibility threshold. When $U \leq T+1$, secure aggregation in the information-theoretic sense is infeasible. Intuitively, in this regime, the number of potentially colluding users is too large relative to the minimum number of users required for correct decoding. As a result, any scheme that guarantees correctness inevitably exposes additional information beyond the desired sum, making it impossible to simultaneously satisfy the correctness and security requirements.

\tit{2) Optimal Rates:} When $U > T+1$, secure aggregation becomes feasible and admits a particularly structured optimal solution. The optimal communication strategy takes a simple form: each user sends one symbol in the first round, and the second round requires $\frac{1}{U-T-1}$ symbols per input symbol. The first-round rate must be at least one because a user may drop out in the second round, while its input must still be recoverable as part of the desired aggregation. Consequently, sufficient information about each user's input must already be conveyed in the first-round transmission. This shows that the communication cost is determined by the effective redundancy level $U-T-1$, which measures the remaining degree of cooperation after accounting for worst-case collusion. A larger value of $U-T-1$ implies more available redundancy among active users, and therefore directly reduces the communication overhead in the second round.

Finally, we note that the optimal communication rate does not depend on the number of users $K$. This indicates that the scalability of the system does not affect the fundamental communication efficiency; instead, feasibility and efficiency are determined solely by the relative values of $U$ and $T$.

\section{Achievability Proof of Theorem \ref{thm:main}}\label{sec:ach}

We now present the achievability proof of Theorem~\ref{thm:main}. The proof is first illustrated through two representative examples, which highlight the main ingredients of the construction in a transparent setting. Based on these examples, we then generalize the scheme to arbitrary system parameters, followed by the analysis of correctness, security, and achievable rates.

\subsection{Example~1: $K = 4, U = 3, T = 0$}

We illustrate the main ideas using a simple instance of the DSA problem with $K = 4$ users, where at least $U = 3$ users survive and there are no colluding users, i.e., $T=0$. We choose the input length as $L = U - T -1 = 2$. Accordingly, each user input is given by
$W_{k} = \big(W_{k}(1), W_{k}(2)\big)^{\top} \in \mathbb{F}_q^{2 \times 1}$.
Next, we define the correlated keys. Each user $k\in [4]$ has two independent and uniformly distributed random vectors over $\mathbb{F}_q$: 
$N_{k} = \bigl(N_{k}(1), N_{k}(2)\bigr)^{\top} \in \mathbb{F}_q^{2 \times 1}$,
and
$S_{k}  \in \mathbb{F}_q$.
Based on these local random variables, we construct additional shared keys through linear combinations. Specifically, for each user $k$,  the secret keys  are given by
\begin{eqnarray}
Z_{k} = \Big( N_{k}, \big\{ (N_{i}(1), N_{i}(2), S_{i}) \cdot \boldsymbol{\alpha}_{k} \big\}_{i\in[4]} \Big). \label{ex1keys}
\end{eqnarray}

Before specifying the coefficient vectors $\boldsymbol{\alpha}_{k}$, we introduce the maximum distance separable (MDS) matrix used in the construction.
A matrix $\boldsymbol{\alpha} \in \mathbb{F}_q^{3 \times 4}$ is called an MDS matrix if any $3 \times 3$ submatrix is invertible over $\mathbb{F}_q$, i.e., any three columns are linearly independent. This property ensures that, given any three linear combinations corresponding to a subset of users of size three, each user can uniquely recover the underlying keys.

We construct $\boldsymbol{\alpha}$ as a Vandermonde matrix over $\mathbb{F}_q$,
\[
\boldsymbol{\alpha} = 
\begin{pmatrix}
1 & \beta_1 & \beta_1^2 & \beta_1^3 \\
1 & \beta_2 & \beta_2^2 & \beta_2^3  \\
1 & \beta_3 & \beta_3^2 & \beta_3^3  
\end{pmatrix},
\]
where $\{\beta_1, \beta_2, \beta_3\} \subset \mathbb{F}_q$ are distinct nonzero elements.

Then, for each user $k$, the coefficient vector $\boldsymbol{\alpha}_{k} \in \mathbb{F}_q^{3 \times 1}$ is taken as the $k$-th column of $\boldsymbol{\alpha}$. The inner product of a 3-dimensional row vector with $\boldsymbol{\alpha}_{k}$ produces a scalar used in the linear combination for each $k$.

\textbf{Explicit instantiation:} We can choose
$\boldsymbol{\alpha}_{k} = \bigl(1, 2^{k-1}, 3^{k-1}\bigr)^\top,$
provided that the multiplicative orders of $2,3,4$ in $\mathbb{F}_q$ are all greater than $k$, so that the powers remain distinct. In particular, $q$ must be an odd prime or odd prime power, i.e., $q \neq 2^m$, to avoid degeneracy in characteristic 2.
This construction ensures that each user has access to both private and globally mixed linear combinations of the keys, which is essential for achieving the desired security and recoverability properties.
For example, for each user $k\in [4]$, the keys are
\begin{align}
Z_1=\Big(N_1,  &N_{1}(1)+ N_{1}(2)+S_{1},\notag\\
    &N_{2}(1)+ N_{2}(2)+S_{2},\notag\\
    &N_{3}(1)+ N_{3}(2)+S_{3},\notag\\
    &N_{4}(1)+ N_{4}(2)+S_{4}\Big),\notag\\
    Z_2=\Big(N_2,  &N_{1}(1)+ 2N_{1}(2)+3S_{1},\notag\\
    &N_{2}(1)+ 2N_{2}(2)+3S_{2},\notag\\
    &N_{3}(1)+ 2N_{3}(2)+3S_{3},\notag\\
    &N_{4}(1)+ 2N_{4}(2)+3S_{4}\Big),\notag\\
    Z_3=\Big(N_3,  &N_{1}(1)+ 2^2N_{1}(2)+3^2S_{1},\notag\\
    &N_{2}(1)+ 2^2N_{2}(2)+3^2S_{2},\notag\\
    &N_{3}(1)+ 2^2N_{3}(2)+3^2S_{3},\notag\\
    &N_{4}(1)+ 2^2N_{4}(2)+3^2S_{4}\Big),\notag\\
    Z_4=\Big(N_4,  &N_{1}(1)+ 2^3N_{1}(2)+3^3S_{1},\notag\\
    &N_{2}(1)+ 2^3N_{2}(2)+3S_{2},\notag\\
    &N_{3}(1)+ 2^3N_{3}(2)+3^3S_{3},\notag\\
    &N_{4}(1)+ 2^3N_{4}(2)+3^3S_{4}\Big).
\end{align}

We have now completed the design of the correlated secret keys. We next describe how the messages are constructed over the two communication rounds.

\textbf{First round:}
Each user broadcasts the message $X_k$
\begin{align}
X_{k} = W_{k} + N_{k},
\quad \forall k\in[4]. \label{ex1X1}
\end{align}

After the first round, based on the received messages, each user can determine the set of surviving users $\mathcal{U}_1$. Then, each user proceeds to the second-round transmission accordingly.

\textbf{Second round:}
For each surviving user $k\in\mathcal{U}_1$, we set
\begin{align}
Y^{\mathcal{U}_1}_{k}
=& \sum_{i\in\mathcal{U}_1}
\Big( (N_{i}(1), N_{i}(2), S_{i}) \cdot \boldsymbol{\alpha}_{k} 
\Big)\notag\\
=& \sum_{i\in\mathcal{U}_1}
\Big(
N_{i}(1)
+2^{k-1}N_{i}(2)
+3^{k-1}S_{i}
\Big).
\end{align}

For example, if $\mathcal{U}_1=\{1,2,4\}$, we have
\begin{align}
    Y^{\mathcal{U}_1}_1=&N_{1}(1)+N_{1}(2)+S_{1}+\notag\\
    &N_{2}(1)+N_{2}(2)+S_{2}+\notag\\
    &N_{4}(1)+N_{4}(2)+S_{4},\notag\\
    Y^{\mathcal{U}_1}_2=&N_{1}(1)+2N_{1}(2)+3S_{1}+\notag\\
    &N_{2}(1)+2N_{2}(2)+3S_{2}+\notag\\
    &N_{4}(1)+2N_{4}(2)+3S_{4},\notag\\
    Y^{\mathcal{U}_1}_4=&N_{1}(1)+2^3N_{1}(2)+3^3S_{1}+\notag\\
    &N_{2}(1)+2^3N_{2}(2)+3^3S_{2}+\notag\\
    &N_{4}(1)+2^3N_{4}(2)+3^3S_{4}.
\end{align}

We next analyze the achievable rates and prove correctness and security.

\emph{Rate:}
Finally, we specify the communication rates of the proposed scheme.
The first-round rate is given by
$R_X=\frac{L_X}{L}=1.$
For the second round, we have
$R_Y=\frac{L_Y}{L}=\frac{1}{2}.$
Equivalently, the second-round rate can be expressed in terms of the system
parameters as
$R_Y=\frac{1}{U-T-1}.$
With these rates, the proposed scheme satisfies both the correctness and the
security requirements, as shown below.

\emph{Correctness:}
Let $\mathcal{U}_1=\{1,2,4\}$ denote the set of surviving users after the first round, and let $\mathcal{U}_2=\{1,2,4\}$ denote the set of surviving users in the second round. Since $U=3$, it follows that $|\mathcal{U}_1|\ge U$ and $|\mathcal{U}_2|\ge U$, ensuring that at least $U$ users survive in both rounds.
From the second-round messages $\{Y^{\mathcal{U}_1}_k\}_{k\in \mathcal{U}_2}$, each user can recover the aggregation keys
$\sum_{k\in \mathcal{U}_1} N_{k} =$ $ \bigl(\sum_{k\in \mathcal{U}_1} N_{k}(1),$ $ \sum_{k\in \mathcal{U}_1} N_{k}(2)\bigr)$ and $\sum_{k\in \mathcal{U}_1} S_{k}$
with zero error, since the precoding matrices in (\ref{ex1keys}) are MDS.
Combining this with the sum of the first-round messages,
$X_1 + X_2 + X_4 = \sum_{k\in \mathcal{U}_1} \bigl(W_{k} + N_{k}\bigr),$
each user can cancel $\sum_{k\in \mathcal{U}_1} N_k$ and successfully decode the desired sum
$\sum_{k\in \mathcal{U}_1} W_{k}$.

{\em Security:} 
When $\mathcal{U}_1=\mathcal{U}_2=\{1,2,4\}$,
we verify that the security constraint~\eqref{security} is satisfied.
Consider user~1, we have
\begin{align}
    &I\Big(X_{2},X_{3},X_{4},Y^{\mathcal{U}_1}_2,Y^{\mathcal{U}_1}_4;\{W_{k}\}_{k\in [4]}\Big|\notag\\
    &W_{1}+W_{2}+W_{4},W_{1},Z_{1}\Big)\\
    =&H\Big(X_{2},X_{3},X_{4},Y^{\mathcal{U}_1}_2,Y^{\mathcal{U}_1}_4\Big|W_{1}+W_{2}+W_{4},W_{1},Z_{1}\Big)\notag\\
    &-H\Big(X_{2},X_{3},X_{4},Y^{\mathcal{U}_1}_2,Y^{\mathcal{U}_1}_4\Big|\{W_{k}\}_{k\in [4]},Z_{1}\Big)\\
    =&H\Big(X_{2},X_{3},X_{4},Y^{\mathcal{U}_1}_2,Y^{\mathcal{U}_1}_4,W_{1}+W_{2}+W_{4},W_{1},Z_{1}\Big)\notag\\
    &-H\Big(W_{1}+W_{2}+W_{4},W_{1},Z_{1}\Big)\notag\\
    &-H\Big(N_{2},N_{3},N_{4},Y^{\mathcal{U}_1}_2,Y^{\mathcal{U}_1}_4\Big|\{W_{k}\}_{k\in [4]},Z_{1}\Big)\label{secpfex1t1}\\
    =&H\Big(X_{2},X_{3},X_{4},Y^{\mathcal{U}_1}_2,Y^{\mathcal{U}_1}_4,W_{1},Z_{1}\Big)\notag\\
    &-H\Big(W_{1}+W_{2}+W_{4},W_{1},Z_{1}\Big)\notag\\
    &-\Big(H(N_{2},N_{3},N_{4},Y^{\mathcal{U}_1}_2,Y^{\mathcal{U}_1}_4,Z_{1})-H(Z_1)\Big)\label{secpfex1t11}\\
    =&H\Big(X_{2},X_{3},X_{4},Y^{\mathcal{U}_1}_2,Y^{\mathcal{U}_1}_4,W_{1},Z_{1}\Big)\notag\\
    &-H(W_{1}+W_{2}+W_{4},W_{1},Z_{1})\notag\\
    &-(H(N_1,N_2,N_3,N_4,S_1,S_2,S_3,S_4)-H(Z_1))\label{secpfex1t2}\\
    =&16-10-(12-6)=0.
\end{align}
The equality in \eqref{secpfex1t1} follows from substituting $X_k=W_k+Z_k$ in (\ref{ex1X1}), which removes the dependence on $\{W_k\}_{k\in[4]}$ and reduces the random variables to $(N_2,N_3,N_4,Y^{\mathcal{U}_1}_2,Y^{\mathcal{U}_1}_4)$.
In \eqref{secpfex1t11}, the first term follows since $W_{1}+W_{2}+W_{4}$ can be expressed as a deterministic function of $X_{2},X_{4},Y^{\mathcal{U}_1}_2,Y^{\mathcal{U}_1}_4,W_{1},Z_{1}$.
The equality in \eqref{secpfex1t2} follows from the structure of the key design in (\ref{ex1keys}). In particular, given $(N_2,N_3,N_4,Z_1)$, the full keys $(N_1,N_2,N_3,N_4,S_1,S_2,S_3,S_4)$ can be reconstructed.
Finally, substituting the entropy values yields zero mutual information, which implies that no additional information about $\{W_k\}_{k\in[4]}$ is leaked. Hence, the security constraint is satisfied.

\subsection{Example~2: $K = 4, U = 3, T = 1$}

We consider the DSA problem with $K = 4$ users, where at least $U = 3$ users survive and there is $T=1$ colluding user. We choose the input length as $L = U - T -1 = 1$.
Accordingly, each user input is given by $W_{k} \in \mathbb{F}_q$.
Next, we define the correlated keys. Each user $k\in [4]$ has two independent and uniformly distributed random variables over $\mathbb{F}_q$: $N_{k} \in \mathbb{F}_q$,
and
$S_{k}  = \bigl(S_{k}(1), S_{k}(2)\bigr)^{\top} \in \mathbb{F}_q^{2 \times 1}$.
Based on these local random variables, we construct $K = 4$ linearly independent linear combinations, which are used to encode additional keys accessible to each user. Specifically, for each user $k$, the available keys is defined as
\begin{eqnarray}
Z_{k} = \Big( N_{k}, \big\{ (N_{i}, S_{i}(1), S_{i}(2)) \cdot \boldsymbol{\alpha}_{k} \big\}_{i\in[4]} \Big). \label{ex2keys}
\end{eqnarray}

Before specifying the coefficient vectors $\boldsymbol{\alpha}_{k}$, we define the MDS matrix used in the construction.

A matrix $\boldsymbol{\alpha} \in \mathbb{F}_q^{3 \times 4}$ is called an MDS matrix if any $3 \times 3$ submatrix is invertible over $\mathbb{F}_q$, i.e., any three columns are linearly independent. This guarantees that any three linear combinations determine the remaining keys components.

We construct $\boldsymbol{\alpha}$ as a Vandermonde matrix over $\mathbb{F}_q$,
\[
\boldsymbol{\alpha} = 
\begin{pmatrix}
1 & \beta_1 & \beta_1^2 & \beta_1^3 \\
1 & \beta_2 & \beta_2^2 & \beta_2^3  \\
1 & \beta_3 & \beta_3^2 & \beta_3^3  
\end{pmatrix},
\]
where $\{\beta_1, \beta_2, \beta_3\} \subset \mathbb{F}_q$ are distinct nonzero elements.

Then, for each user $k$, the coefficient vector $\boldsymbol{\alpha}_{k} \in \mathbb{F}_q^{4 \times 1}$ is taken as the $k$-th column of $\boldsymbol{\alpha}$.

\textbf{Explicit instantiation:} We can choose
\[
\boldsymbol{\alpha}_{k} = \bigl(1, 2^{k-1}, 3^{k-1}\bigr)^\top,
\]
provided that the multiplicative orders of $2,3,4$ in $\mathbb{F}_q$ are sufficiently large so that the powers remain distinct. In particular, $q$ is chosen as an odd prime or odd prime power, i.e., $q \neq 2^m$, to avoid degeneracy in characteristic 2.
This construction ensures that each user has access to both private keys and globally mixed linear combinations, which are essential for security and recoverability.
For example, the keys of the users are 
\begin{align}
    Z_1=\Big(N_1,  &N_{1}+ S_{1}(1)+S_{1}(2),\notag\\
    &N_{2}+ S_{2}(1)+S_{2}(2),\notag\\
    &N_{3}+ S_{3}(1)+S_{3}(2),\notag\\
    &N_{4}+ S_{4}(1)+S_{4}(2)\Big),\notag\\
    Z_2=\Big(N_2,  &N_{1}+ 2S_{1}(1)+3S_{1}(2),\notag\\
    &N_{2}+ 2S_{2}(1)+3S_{2}(2),\notag\\
    &N_{3}+ 2S_{3}(1)+3S_{3}(2),\notag\\
    &N_{4}+ 2S_{4}(1)+3S_{4}(2)\Big),\notag\\
    Z_3=\Big(N_3,  &N_{1}+ 2^2S_{1}(1)+3^2S_{1}(2),\notag\\
    &N_{2}+ 2^2S_{2}(1)+3^2S_{2}(2),\notag\\
    &N_{3}+ 2^2S_{3}(1)+3^2S_{3}(2),\notag\\
    &N_{4}+ 2^2S_{4}(1)+3^2S_{4}(2)\Big),\notag\\
    Z_4=\Big(N_4,  &N_{1}+ 2^3S_{1}(1)+3^3S_{1}(2),\notag\\
    &N_{2}+ 2^3S_{2}(1)+3^3S_{2}(2),\notag\\
    &N_{3}+ 2^3S_{3}(1)+3^3S_{3}(2),\notag\\
    &N_{4}+ 2^3S_{4}(1)+3^3S_{4}(2)\Big).
\end{align}

We now describe the communication scheme.

\textbf{First round:}
Each user broadcasts the message
\begin{align}
X_{k} = W_{k} + N_{k},
\quad \forall k\in[4]. \label{ex2X1}
\end{align}
After the first round, each user can determine the surviving set $\mathcal{U}_1$. Based on $\mathcal{U}_1$, users proceed to the second round.

\textbf{Second round:}
For each surviving user $k\in\mathcal{U}_1$, we set
\begin{align}
Y^{\mathcal{U}_1}_{k}
=& \sum_{i\in\mathcal{U}_1}
\Big( (N_{i}, S_{i}(1), S_{i}(2)) \cdot \boldsymbol{\alpha}_{k} \Big)\notag\\
=& \sum_{i\in\mathcal{U}_1}
\Big(
N_{i}
+2^{k-1}S_{i}(1)
+3^{k-1}S_{i}(2)
\Big).
\end{align}
For example, if $\mathcal{U}_1=\{1,2,4\}$, we have
\begin{align}
    Y^{\mathcal{U}_1}_1=&N_{1}+S_{1}(1)+S_{1}(2)+\notag\\
    &N_{2}+S_{2}(1)+S_{2}(2)+\notag\\
    &N_{4}+S_{4}(1)+S_{4}(2),\notag\\
    Y^{\mathcal{U}_1}_2=&N_{1}+2S_{1}(1)+3S_{1}(2)+\notag\\
    &N_{2}+2S_{2}(1)+3S_{2}(2)+\notag\\
    &N_{4}+2S_{4}(1)+3S_{4}(2),\notag\\
    Y^{\mathcal{U}_1}_4=&N_{1}+2^3S_{1}(1)+3^3S_{1}(2)+\notag\\
    &N_{2}+2^3S_{2}(1)+3^3S_{2}(2)+\notag\\
    &N_{4}+2^3S_{4}(1)+3^3S_{4}(2).
\end{align}

\emph{Rate:}
The first-round rate is
$R_X=1.$
The second-round rate is
$R_Y=1,$
which is equivalently expressed as
$R_Y=\frac{1}{U-T-1}.$
These rates satisfy the requirements of the scheme.

We next prove correctness and security of the above \schm.

\emph{Correctness:}
Let $\mathcal{U}_1=\{1,2,4\}$ and $\mathcal{U}_2=\{1,2,4\}$. Since $U=3$, both sets satisfy $|\mathcal{U}_1|,|\mathcal{U}_2|\ge U$.
From $\{Y^{\mathcal{U}_1}_k\}_{k\in \mathcal{U}_2}$, each user can recover
$\sum_{k\in \mathcal{U}_1} N_{k}$ and $\sum_{k\in \mathcal{U}_1} S_{k}=(\sum_{k\in \mathcal{U}_1} S_{k}(1),\sum_{k\in \mathcal{U}_1} S_{k}(2))$,
due to the MDS property of $\boldsymbol{\alpha}$.
Combining this with
$X_1 + X_2 + X_4 = \sum_{k\in \mathcal{U}_1} (W_k + N_k),$
each user can cancel $\sum_{k\in \mathcal{U}_1} N_k$ and recover $\sum_{k\in \mathcal{U}_1} W_k$.

{\em Security:} 
When $\mathcal{U}_1=\mathcal{U}_2=\{1,2,4\}$,
we show that the security constraint~\eqref{security} is satisfied.
Consider that user~1 colludes with user~3, we have
\begin{align}
    &I\Big(X_{2},X_{3},X_{4},Y^{\mathcal{U}_1}_2,Y^{\mathcal{U}_1}_4;\{W_{k}\}_{k\in [4]}\Big|\notag\\
    &W_{1}+W_{2}+W_{4},W_{1},Z_{1},W_{3},Z_{3}\Big)\\
    =&H\Big(X_{2},X_{3},X_{4},Y^{\mathcal{U}_1}_2,Y^{\mathcal{U}_1}_4\Big|W_{1}+W_{2}+W_{4},W_{1},Z_{1},W_{3},\notag\\
    &Z_{3}\Big)-H\Big(X_{2},X_{3},X_{4},Y^{\mathcal{U}_1}_2,Y^{\mathcal{U}_1}_4|\{W_{k}\}_{k\in [4]},Z_{1},Z_{3}\Big)\\
    =&H\Big(X_{2},X_{3},X_{4},Y^{\mathcal{U}_1}_2,Y^{\mathcal{U}_1}_4,W_{1}+W_{2}+W_{4},W_{1},Z_{1},W_{3},\notag\\
    &Z_{3}\Big)-H(W_{1}+W_{2}+W_{4},W_{1},Z_{1},W_{3},Z_{3})\notag\\
    &-H\Big(N_{2},N_{3},N_{4},Y^{\mathcal{U}_1}_2,Y^{\mathcal{U}_1}_4\Big|\{W_{k}\}_{k\in [4]},Z_{1},Z_3\Big)\label{secpfex2t1}\\
    =&H\Big(X_{2},X_{4},Y^{\mathcal{U}_1}_2,W_{1},Z_{1},W_3,Z_3\Big)\notag\\
    &-H(W_{1}+W_{2}+W_{4},W_{1},Z_{1},W_3,Z_3)\notag\\
    &-\Big(H\Big(N_{2},N_{4},Y^{\mathcal{U}_1}_2,Y^{\mathcal{U}_1}_4,Z_{1},Z_3\Big)-H(Z_1,Z_3)\Big)\label{secpfex2t11}\\
    =&H\Big(X_{2},X_{4},Y^{\mathcal{U}_1}_2,W_{1},Z_{1},W_{3},Z_{3}\Big)\notag\\
    &-H(W_{1}+W_{2}+W_{4},W_{1},Z_{1},W_{3},Z_{3})\notag\\
    &-(H(N_1,N_2,N_3,N_4,S_1,S_2,S_3,S_4)-H(Z_1,Z_3))\label{secpfex2t2}\\
    =&15-13-(12-10)=0.
\end{align}
The equality in \eqref{secpfex2t1} follows by substituting $X_k=W_k+N_k$ as defined in~\eqref{ex2X1}, which removes the dependence on $\{W_k\}_{k\in[4]}$ and reduces the remaining random variables to $(N_2,N_3,N_4,Y^{\mathcal{U}_1}_2,Y^{\mathcal{U}_1}_4)$.
The equality in \eqref{secpfex2t11} holds since $W_{1}+W_{2}+W_{4}$ can be determined from $(X_{2},X_{4},Y^{\mathcal{U}_1}_2,Y^{\mathcal{U}_1}_4,W_{1},Z_{1})$. Moreover, $Y^{\mathcal{U}_1}_4$ is functionally related to $(Y^{\mathcal{U}_1}_2, Z_1, Z_3)$ due to the MDS property of $\boldsymbol{\alpha} \in \mathbb{F}_q^{3 \times 4}$, where any three columns determine the remaining ones.
The equality in \eqref{secpfex2t2} follows from the structure of the key generation scheme in~\eqref{ex2keys}. In particular, given $(N_2,N_4,Z_1,Z_3)$, the full keys tuple $(N_1,N_2,N_3,N_4,S_1,S_2,S_3,S_4)$ can be reconstructed.
Finally, evaluating the entropy terms yields zero mutual information, implying that no additional information about $\{W_k\}_{k\in[4]}$ is leaked. Hence, the security constraint is satisfied.

We now extend the above example to a general scheme applicable to arbitrary system parameters, while satisfying the correctness and security constraints.

\subsection{General Scheme for Arbitrary $K, U, T$}

We next present a general achievable scheme for the feasible
regime $U>T+1$. Before describing the construction, we first
briefly explain why the condition $U>T+1$ is 
necessary.

Intuitively, when the number of surviving users is not
strictly larger than the collusion level plus one, the users
participating in decoding do not provide sufficient redundancy
to simultaneously guarantee correctness and information-theoretic
security. In particular, correctness requires that the surviving
users collectively recover the desired aggregation from the
broadcast messages, while security requires that any user
colluding with at most $T$ other users should learn nothing
beyond the target sum. However, when $U \le T+1$, the set
of users involved in decoding can already be covered by the
colluding users together with the observing user. Consequently,
the information required for correct decoding inevitably leaks
additional information about individual users' inputs, making the correctness and security constraints 
incompatible.

This intuition will be rigorously formalized in the converse proof of Section~\ref{inf}. Therefore, in the remainder of
this section, we focus on the feasible regime $U>T+1$ and
present an achievable scheme attaining the optimal rates.

Consider the DSA problem with
$K$ users, where at least $U$ users survive. Moreover, each
user $k \in [K]$ may collude with at most $T$ other users.
Accordingly, we choose the input length as
$
L = U - T - 1.
$
Each user $k$ holds an input vector
$
W_{k} = \bigl(W_{k}(1), \cdots, W_{k}(U - T - 1)\bigr)^\top
\in \mathbb{F}_q^{(U - T - 1)\times 1}.
$

We next specify the correlated secret keys used in the scheme.
Consider a total of $2K$ independent and uniformly distributed random vectors over $\mathbb{F}_q$, given by
$N_{k} = \bigl(N_{k}(1), \cdots, N_{k}(U - T -1)\bigr)^\top \in \mathbb{F}_q^{(U - T - 1)\times 1},$ and $S_{k} = \bigl(S_{k}(1), \cdots, S_{k}(T+1)\bigr)^\top \in \mathbb{F}_q^{(T+1)\times 1},$ for all $k\in[K]$.
From these secret keys, we construct a collection of linearly independent combinations over $\mathbb{F}_q$, with coefficients determined by a carefully designed MDS matrix. Specifically, for each $k\in[K]$, the keys available to user $k$ is defined as
\begin{equation}
Z_{k} =\left(N_{k},\bigl\{[\boldsymbol{Q}_{i}]_{k}\bigr\}_{i\in[K]}\right).\label{eq:zz}
\end{equation}
Here,
$[\boldsymbol{Q}_{i}]_{k}\triangleq \Bigl(N_{i}(1), \cdots, N_{i}(U - T - 1),S_{i}(1), \cdots, S_{i}(T+1)\Bigr)\boldsymbol{\alpha}_{k}\in \mathbb{F}_q,
$ where $\boldsymbol{\alpha}_{k}\in\mathbb{F}_q^{U\times 1}$ denotes the $k$-th column of a matrix
$\boldsymbol{\alpha}\in\mathbb{F}_q^{U\times K}$.
We say that a matrix $\boldsymbol{\alpha}\in\mathbb{F}_q^{U\times K}$ with $U < K$ is an MDS matrix if any
$U\times U$ submatrix is nonsingular. Furthermore, $\boldsymbol{\alpha}$ is said to be a $(T+1)$-private MDS matrix if the submatrix formed by its last $T+1$ rows is also MDS. Such a matrix exists for sufficiently large field size $q$.
A $(T+1)$-private MDS matrix~\cite{so2022lightsecagg} guarantees that, for any subset
$\mathcal{T}\cup\{u\}\subseteq[K]$ with $|\mathcal{T}\cup\{u\}|\le T+1$, the collection of linear projections
$\{[\boldsymbol{Q}_{i}]_{k}\}_{i\in[K],\,k\in\mathcal{T}\cup\{u\}}$
is statistically independent of the masking variables
$\{N_{i}\}_{i\in[K]}$, i.e.,
\begin{equation}
I\left(
\{[\boldsymbol{Q}_{i}]_{k}\}_{i\in[K],\,k\in(\mathcal{T}\cup\{u\})};
\{N_{i}\}_{i\in[K]}
\right)=0.
\label{Tprivacy}
\end{equation}
This follows from the fact that the submatrix formed by the last $T+1$ rows of
$\boldsymbol{\alpha}$ is MDS, which ensures that any set of at most $T+1$ such projections depends only on the keys vectors $\{S_{i}\}$ and reveals no information about $\{N_{i}\}$.

To facilitate the subsequent security analysis, we next summarize several useful entropy properties of the secret keys in the following lemma.

\begin{lemma}
For any subset $\mathcal{T}\cup\{u\}\subseteq[K]$ with $|\mathcal{T}\cup\{u\}|\le T+1$,
any $\mathcal{U}_2\subseteq \mathcal{U}_1 \subseteq [K]$
satisfying $|\mathcal{U}_2|\ge U$, with $U>T+1$, we have
\begin{equation}
I\left( \{N_{i}\}_{i\in[K]\setminus(\mathcal{T}\cup\{u\})}; \{N_{k},\{[\boldsymbol{Q}_{i}]_{k}\}_{i\in[K]}\}_{k\in(\mathcal{T}\cup\{u\})} \right)=0.
\label{Tprivacy1}
\end{equation}
\end{lemma}

\begin{IEEEproof}
Since $|\mathcal{T}\cup\{u\}|\le T+1$, by the $(T+1)$-privacy property in~\eqref{Tprivacy}, we have
\begin{align}
    &I\left( \{N_{i}\}_{i\in[K]\setminus(\mathcal{T}\cup\{u\})}; \{N_{k},\{[\boldsymbol{Q}_{i}]_{k}\}_{i\in[K]}\}_{k\in(\mathcal{T}\cup\{u\})} \right)\notag\\
    =&H\left(\{N_{k},\{[\boldsymbol{Q}_{i}]_{k}\}_{i\in[K]}\}_{k\in(\mathcal{T}\cup\{u\})}\right)-\notag\\
    &H\left(\{N_{k},\{[\boldsymbol{Q}_{i}]_{k}\}_{i\in[K]}\}_{k\in(\mathcal{T}\cup\{u\})}\,\middle|\, \{N_{i}\}_{i\in[K]\setminus(\mathcal{T}\cup\{u\})}\right)\notag\\
    =&H\left(\{N_{k},\{[\boldsymbol{Q}_{i}]_{k}\}_{i\in[K]}\}_{k\in(\mathcal{T}\cup\{u\})}\right)\notag\\
    &-H\left(\{N_{k}\}_{k\in(\mathcal{T}\cup\{u\})}\,\middle|\, \{N_{i}\}_{i\in[K]\setminus(\mathcal{T}\cup\{u\})}\right)\notag\\
    &-H\left(\{\{[\boldsymbol{Q}_{i}]_{k}\}_{i\in[K]}\}_{k\in(\mathcal{T}\cup\{u\})}\,\middle|\, \{N_{i}\}_{i\in[K]}\right)\\
    =&H\left(\{N_{k},\{[\boldsymbol{Q}_{i}]_{k}\}_{i\in[K]}\}_{k\in(\mathcal{T}\cup\{u\})}\right)\notag\\
    &
    -H\left(\{N_{k}\}_{k\in(\mathcal{T}\cup\{u\})}\right)-H\left(\{\{[\boldsymbol{Q}_{i}]_{k}\}_{i\in[K]}\}_{k\in(\mathcal{T}\cup\{u\})}\right)\notag\\
    &+\underbrace{I\left(\{\{[\boldsymbol{Q}_{i}]_{k}\}_{i\in[K]}\}_{k\in(\mathcal{T}\cup\{u\})}; \{N_{i}\}_{i\in[K]}\right)}_{\overset{(\ref{Tprivacy})}{=}0}\label{lemma1t1}\\
    \leq &H\left(\{N_{k}\}_{k\in(\mathcal{T}\cup\{u\})}\right)
    +H\left(\{\{[\boldsymbol{Q}_{i}]_{k}\}_{i\in[K]}\}_{k\in(\mathcal{T}\cup\{u\})}\right)\notag\\
    &-H\left(\{N_{k}\}_{k\in(\mathcal{T}\cup\{u\})}\right)\notag\\
    &
    -H\left(\{\{[\boldsymbol{Q}_{i}]_{k}\}_{i\in[K]}\}_{k\in(\mathcal{T}\cup\{u\})}\right)\\
    =0.
\end{align}
where the fourth term in~\eqref{lemma1t1} equals zero by~\eqref{Tprivacy}, which follows from the $(T+1)$-privacy property of the MDS matrix $\boldsymbol{\alpha}$.
\end{IEEEproof}

With the correlated secret keys fully specified, we proceed to describe the message transmissions over two rounds.

\textbf{First round:}
Each user transmits
\begin{align}
X_{k} = W_{k} + N_{k}, \quad \forall k\in[K].
\label{r1h1}
\end{align}

After the first round, each user determines the set of surviving users $\mathcal{U}_1$ based on the received messages, and subsequently requests the transmission of the second-round messages.

\textbf{Second round:}
To enable the recovery of the aggregated masking variables at each user, each surviving user $k\in\mathcal{U}_1$ transmits
\begin{align}
Y^{\mathcal{U}_1}_{k} = \sum_{i\in\mathcal{U}_1} [\boldsymbol{Q}_{i}]_{k}.
\label{r2h1}
\end{align}

With the correlated keys and two-round message transmissions fully specified, we proceed to analyze the scheme's achievable rate, correctness, and information-theoretic security.






\emph{Rate:}
Each user input has length $L = U - T - 1$ symbols over $\mathbb{F}_q$. 
In the first round, each user transmits $L_X = U - T -1$ symbols, yielding
$R_X = \frac{L_X}{L} = 1.$ In the second round, each user transmits a single symbol, i.e., $L_Y = 1$, which gives $R_Y = \frac{L_Y}{L} = \frac{1}{U - T - 1}.$ Therefore, the proposed scheme achieves the rate pair $\left(1,\frac{1}{U-T-1}\right),$ while satisfying both the correctness and security requirements.

\emph{Correctness:}  
Each user $k\in \mathcal{U}_2$ receives the second-round messages $\{Y^{\mathcal{U}_1}_k\}_{k\in \mathcal{U}_2}$, which, by~\eqref{r2h1}, equal
$\left\{\sum_{i\in\mathcal{U}_1} [\boldsymbol{Q}_{i}]_{k}\right\}_{k\in \mathcal{U}_2}$.
Each encoded symbol $[\boldsymbol{Q}_{i}]_{k}$ is defined as
$
[\boldsymbol{Q}_{i}]_{k}
\triangleq
\bigl(N_{i}(1), \ldots, N_{i}(U - T - 1), S_{i}(1), \ldots, S_{i}(T+1)\bigr)\boldsymbol{\alpha}_{k},
$
where $\boldsymbol{\alpha} = [\boldsymbol{\alpha}_{k}]_{k\in [K]} \in \mathbb{F}_q^{U \times K}$ is an MDS matrix.
By the MDS property, any $U$ columns of $\boldsymbol{\alpha}$ are linearly independent. Since each user $u\in \mathcal{U}_2$ receives messages from at least $|\mathcal{U}_2|\ge U$ users, it obtains at least $U$ coded symbols of the form
$\sum_{i\in\mathcal{U}_1} [\boldsymbol{Q}_{i}]_{k}$.
These correspond to $U$ linearly independent columns of $\boldsymbol{\alpha}$, which enables recovery of the aggregated keys vector
$
\sum_{i\in \mathcal{U}_1}
\bigl(N_{i}(1), \ldots, N_{i}(U - T - 1), S_{i}(1), \ldots, S_{i}(T+1)\bigr).
$
Equivalently, each user can recover $\sum_{i\in \mathcal{U}_1} N_{i}$ and $\sum_{i\in \mathcal{U}_1} S_{i}$ exactly.

Finally, combining the recovered aggregated keys with the sum of the first-round messages,
$
\sum_{k\in \mathcal{U}_1} X_{k}
\overset{\eqref{r1h1}}{=}
\sum_{k\in \mathcal{U}_1} (W_{k} + N_{k}),
$
each user can subtract $\sum_{k\in \mathcal{U}_1} N_{k}$ and thus uniquely recover the desired aggregation
$
\sum_{k\in \mathcal{U}_1} W_{k}
$
with zero decoding error.

\emph{Security:}  
Each user $k\in \mathcal{U}_2$ collects messages from other users.
The proposed scheme guarantees that, even if a user colludes with any set of at most $T$ users, it can only learn the aggregate of the surviving users in the first round, and obtains no additional information about the individual messages of the non-colluding users.
\begin{align}
&I\Bigg(\{W_k\}_{k \in [K]}; \{X_k\}_{k \in [K]\setminus \{u\}}, \{Y^{\mathcal{U}_1}_k\}_{k \in \mathcal{U}_1 \setminus \{u\}} \Bigg| \notag\\
&\sum_{k \in \mathcal{U}_1} W_k, \{W_k, Z_k\}_{k \in \mathcal{T} \cup \{u\}} \Bigg)\\
=&H\Bigg( \{X_k\}_{k \in [K]\setminus \{u\}}, \{Y^{\mathcal{U}_1}_k\}_{k \in \mathcal{U}_1 \setminus \{u\}} \Bigg|\sum_{k \in \mathcal{U}_1} W_k, \notag\\
& \{W_k, Z_k\}_{k \in \mathcal{T} \cup \{u\}} \Bigg)-H\Bigg( \{X_k\}_{k \in [K]\setminus \{u\}}, \notag\\
&\{Y^{\mathcal{U}_1}_k\}_{k \in \mathcal{U}_1 \setminus \{u\}} | \{W_k\}_{k \in [K]} , \{W_k, Z_k\}_{k \in \mathcal{T} \cup \{u\}} \Bigg)\\
=&H\Bigg( \{W_k+N_k\}_{k \in [K]\setminus \{u\}}, \Bigg\{ \sum_{i\in\mathcal{U}_1} [\boldsymbol{Q}_{i}]_{k}\Bigg\}_{k \in \mathcal{U}_1 } \Bigg|\sum_{k \in \mathcal{U}_1} W_k, \notag\\
& \{W_k, Z_k\}_{k \in \mathcal{T} \cup \{u\}} \Bigg)-H\Bigg( \{W_k+N_k\}_{k \in [K]\setminus \{u\}},\notag\\
& \Bigg\{ \sum_{i\in\mathcal{U}_1} [\boldsymbol{Q}_{i}]_{k}\Bigg\}_{k \in \mathcal{U}_1 } \Bigg|\{W_k\}_{k \in [K]} , \{W_k, Z_k\}_{k \in \mathcal{T} \cup \{u\}} \Bigg)\\
=&H\Bigg( \{W_k+N_k\}_{k \in [K]\setminus (\mathcal{T} \cup \{u\})} \Bigg|\sum_{k \in \mathcal{U}_1} W_k, \notag\\
&\{W_k, Z_k\}_{k \in \mathcal{T} \cup \{u\}} \Bigg)+H\Bigg( \sum_{i\in \mathcal{U}_1} N_{i}, \sum_{i\in \mathcal{U}_1} S_{i} \Bigg| \notag\\
&\{W_k+N_k\}_{k \in [K]\setminus \{u\}},\sum_{k \in \mathcal{U}_1} W_k,  \{W_k, Z_k\}_{k \in \mathcal{T} \cup \{u\}} \Bigg)-\notag\\
&H( \{W_k+N_k\}_{k \in [K]\setminus (\mathcal{T} \cup \{u\})} |\{W_k\}_{k \in [K]} , \{ Z_k\}_{k \in \mathcal{T} \cup \{u\}} )\notag\\
&-H\Bigg( \sum_{i\in \mathcal{U}_1} N_{i}, \sum_{i\in \mathcal{U}_1} S_{i} \Bigg|\{W_k+N_k\}_{k \in [K]\setminus \{u\}}, \notag\\
&\{W_k\}_{k \in [K]} ,  \{W_k, Z_k\}_{k \in \mathcal{T} \cup \{u\}} \Bigg)\label{serversecpft1}\\
\leq&H\Bigg( \{W_k+N_k\}_{k \in [K]\setminus (\mathcal{T} \cup \{u\})}  \Bigg) \notag\\
&+\underbrace{H\Bigg( \sum_{i\in \mathcal{U}_1} N_{i}\Bigg|\{W_k+N_k\}_{k \in \mathcal{U}_1},\sum_{k \in \mathcal{U}_1} W_k\Bigg)}_{{=}0}\notag\\
&+H\Bigg(  \sum_{i\in \mathcal{U}_1} S_{i} \Bigg| \sum_{i\in \mathcal{U}_1} N_{i}, \{W_k+N_k\}_{k \in [K]\setminus \{u\}},\notag\\
&\sum_{k \in \mathcal{U}_1} W_k,  \Bigg\{W_k, N_{k}, \Bigl\{[\boldsymbol{Q}_{i}]_{k}\Bigr\}_{i\in[K]}\Bigg\}_{k \in \mathcal{T} \cup \{u\}} \Bigg)\notag\\
&-H\Bigg( \{N_k\}_{k \in [K]\setminus (\mathcal{T} \cup \{u\})} \Bigg|\notag\\
&\Bigg\{ N_{k},
\Bigl\{[\boldsymbol{Q}_{i}]_{k}\Bigr\}_{i\in[K]}\Bigg\}_{k \in \mathcal{T} \cup \{u\}} \Bigg)\notag\\
&-\underbrace{H\Bigg( \sum_{i\in \mathcal{U}_1} N_{i}\Bigg|\{N_k\}_{k \in [K]}, \{W_k\}_{k \in [K]} ,  \{ Z_k\}_{k \in \mathcal{T} \cup \{u\}} \Bigg)}_{{=}0}\notag\\
&-H\Bigg( \sum_{i\in \mathcal{U}_1} S_{i}\Bigg|\{N_k\}_{k \in [K]}, \{W_k\}_{k \in [K]} ,\notag\\
&\{  N_{k}, \Bigl\{[\boldsymbol{Q}_{i}]_{k}\Bigr\}_{i\in[K]}\}_{k \in \mathcal{T} \cup \{u\}} \Bigg)\label{serversecpft2}\\
=&H\Bigg( \{W_k+N_k\}_{k \in [K]\setminus (\mathcal{T} \cup \{u\})}  \Bigg) \notag\\
&-H( \{N_k\}_{k \in [K]\setminus (\mathcal{T} \cup \{u\})} )+\notag\\
&\underbrace{I\Bigg( \{N_k\}_{k \in [K]\setminus (\mathcal{T} \cup \{u\})} ; \Bigg\{ N_{k},
\Bigl\{[\boldsymbol{Q}_{i}]_{k}\Bigr\}_{i\in[K]}\Bigg\}_{k \in \mathcal{T} \cup \{u\}} \Bigg)}_{\overset{(\ref{Tprivacy1})}{=}0}\label{serversecpft3}\\
=&|[K]\setminus (\mathcal{T} \cup \{u\})|-|[K]\setminus (\mathcal{T} \cup \{u\})|=0.
\end{align}

We next justify the intermediate steps~(\ref{serversecpft1})--(\ref{serversecpft3}).
In~(\ref{serversecpft1}), the second and fourth entropy terms follow from the MDS structure of the second-round encoding. Specifically, the collection $\{Y^{\mathcal{U}_1}_k\}_{k\in(\mathcal{U}_1\setminus\{u\})}$ allows recovery of only the aggregated keys $\sum_{i\in\mathcal{U}_1} N_{i}$ and $\sum_{i\in\mathcal{U}_1} S_{i}$. Given the aggregate $\sum_{k\in\mathcal{U}_1} W_{k}$, these quantities are independent of the individual users' messages and can be separated as shown.
In~(\ref{serversecpft2}), the second and fourth terms are zero since $\sum_{i\in \mathcal{U}_1} N_{i}$ is uniquely determined by $\{W_{k}+N_{k}\}_{k\in\mathcal{U}_1}$ together with $\sum_{k\in\mathcal{U}_1} W_{k}$.
The third term equals the sixth term. When $|\mathcal{T}\cup\{u\}| = T+1$, both terms are zero since $\sum_{i\in \mathcal{U}_1} S_{i}$ is uniquely determined by $\sum_{i\in \mathcal{U}_1} N_{i}$ and the second-round encoding coefficients known to the colluding users. When $|\mathcal{T}\cup\{u\}| < T+1$, the random variable $\sum_{i\in \mathcal{U}_1} S_{i}$ is independent of all conditioned variables, and hence both terms equal $H\!\left(\sum_{i\in \mathcal{U}_1} S_{i}\right)$.
Finally, the third term in~(\ref{serversecpft3}) equals zero due to the $(T+1)$-privacy property in~\eqref{Tprivacy1}, which guarantees that $\{N_k\}_{k \in [K]\setminus (\mathcal{T} \cup \{u\})}$ is independent of the information available to the colluding users.

\section{Converse Proof of Theorem \ref{thm:main}}\label{sec:con}

Before presenting the converse proof, we first establish a basic property that follows from the independence between the inputs $\{W_{k}\}_{k\in[K]}$ and the secret keys $\{Z_{k}\}_{k\in[K]}$, together with the uniform distribution of $\{W_{k}\}_{k\in[K]}$. This property is formalized in the following lemma, which will be repeatedly used in the sequel.

\begin{lemma}\label{lemma:indc}
For any $V_0 < V_1$, $V_0 < V_2$, and $V_1 \neq V_2$, the following holds:
\begin{eqnarray}
I \left(\sum_{k\in[V_1]} W_k ; \sum_{k\in[V_2]} W_k, \left\{ W_k, Z_k \right\}_{k\in[V_0]} \right) = 0. \label{eq:indc}
\end{eqnarray}
\end{lemma}

{\it Proof of Lemma \ref{lemma:indc}:}
\begin{align}
& I \left(\sum_{k\in[V_1]} W_k ; \sum_{k\in[V_2]} W_k, \left\{ W_k, Z_k \right\}_{k\in[V_0]} \right) \notag\\
\overset{(\ref{ind})}{=}& I \left(\sum_{k\in[V_1]} W_k ; \sum_{k\in[V_2]} W_k, \left\{ W_k \right\}_{k\in[V_0]} \right) \\
=& H \left(\sum_{k\in[V_1]} W_k  \right)\notag\\
&-H \left(\sum_{k\in[V_1]} W_k \Bigg| \sum_{k\in[V_2]} W_k, \left\{ W_k \right\}_{k\in[V_0]} \right) \label{lem1t1} \\
=& L - L =0,
\end{align}
where the second term in (\ref{lem1t1}) equals $L$ since $V_1\neq V_2$, $V_0<V_1$, $V_0<V_2$, and $\{W_k\}_{k\in[K]}$ are independent and uniformly distributed.

Building on Lemma~\ref{lemma:indc}, we proceed in several steps. We first characterize the infeasible regime when $U \leq T+1$ in Subsection~\ref{inf}. We then illustrate the main converse ideas through two representative examples in Subsections~\ref{conexpfr1} and~\ref{conexpfr2}. Finally, for the feasible regime $U > T+1$, we derive the converse bounds on the first-round and second-round communication rates in Subsections~\ref{RX1RY1} and~\ref{R2}, respectively.


\subsection{Infeasibility Proof When $U \leq T+1$} \label{inf}

We show that when $U \leq T+1$, the system constraints are contradictory, and hence $\mathcal{R}^* = \emptyset$.  
To establish the contradiction, consider two distinct integers $U_1 \neq U_2$ such that $U < U_2$ and $U < U_1$. Let $\mathcal{U}_1 = [U_1]$ and $\mathcal{T} = [U-1]$, and focus on user $U$, for which $\mathcal{T} \cup \{U\} = [U]$.  
Note that $U \leq T+1 \leq K-2$, so this choice is valid.
\begin{align}
L=&H\Big( \sum_{k\in[U_2]} W_k \Big)\\
    =&I\Bigg( \sum_{k\in[U_2]} W_k; \left\{X_k\right\}_{k\in[U_2]\setminus \{U\}}, \left\{ Y_k^{[U_2]} \right\}_{k\in[U-1]},\notag\\
    &W_U,Z_U \Bigg) +  \underbrace{H\Bigg( \sum_{k\in[U_2]} W_k \Bigg| \left\{X_k\right\}_{k\in[U_2]\setminus \{U\}},}\notag\\
    &\underbrace{\left\{ Y_k^{[U_2]} \right\}_{k\in[U-1]},W_U,Z_U \Bigg)}_{\overset{(\ref{correctness})}{=}0} \label{inft1}\\
    \leq & I\Bigg( \sum_{k\in[U_2]} W_k; \left\{X_k\right\}_{k\in[U_2]\setminus \{U\}}, \left\{ Y_k^{[U_2]} \right\}_{k\in[U-1]},\notag\\
    &\sum_{k\in[U_1]} W_k, \left\{ W_k, Z_k \right\}_{k\in[U]} \Bigg) \\
    = & I\Bigg( \sum_{k\in[U_2]} W_k; \left\{X_k\right\}_{k\in[U_2]\setminus \{U\}}, \left\{ Y_k^{[U_2]} \right\}_{k\in[U-1]}\Bigg|\notag\\
     & \sum_{k\in[U_1]} W_k, \left\{ W_k, Z_k \right\}_{k\in[U]} \Bigg) \notag\\
     &+ \underbrace{I\Bigg( \sum_{k\in[U_2]} W_k; \sum_{k\in[U_1]} W_k, \left\{ W_k, Z_k \right\}_{k\in[U]} \Bigg)}_{=0} \label{inft2}\\
    \overset{(\ref{messageY})}{=} & I\Bigg( \sum_{k\in[U_2]} W_k; \left\{X_k\right\}_{k\in[U_2]\setminus\{U\}}\Bigg| \notag\\
    &\sum_{k\in[U_1]} W_k, \left\{ W_k, Z_k \right\}_{k\in[U]} \Bigg) \label{inft3}\\
    \leq & I\Bigg( \left\{W_k\right\}_{k\in[K]\setminus\{U\}}; \left\{X_k\right\}_{k\in[K]}, \left\{ Y_k^{[U_1]} \right\}_{k\in\mathcal{U}_1\setminus [U]}\Bigg|\notag\\
    &\sum_{k\in[U_1]} W_k, \left\{ W_k, Z_k \right\}_{k\in[U]} \Bigg) \\
    =&0, 
\end{align}
where the second term in (\ref{inft1}) is zero due to the correctness constraint (\ref{correctness}).  
The second term in (\ref{inft2}) is zero by Lemma~\ref{lemma:indc}.  
(\ref{inft3}) follows since the second-round messages $\left\{ Y_k^{[U_2]} \right\}_{k\in[U-1]}$ are deterministic functions of $\left\{ W_k, Z_k \right\}_{k\in[U]}$.  
The last step follows from the security constraint (\ref{security}) with $\mathcal{T}\cup \{U\}=[U]$.

Having established the impossibility result for
$U \leq T+1$, we next focus on the feasible regime
$U>T+1$ and derive converse bounds on the communication
rates.

Before presenting the general converse proof, we first use
Example~1 to illustrate the key converse ideas in a simple
setting. These examples provide useful intuition for the
general arguments developed later. In particular, the first
example explains why the first-round communication must
already contain sufficient information about a user's input
if that user may drop out in the second round, while the
second example illustrates how the security constraint
effectively eliminates part of the communication from
contributing to decoding.

\subsection{Example 1: Converse for First-Round Communication} \label{conexpfr1}

We briefly illustrate the main converse idea using
Example~1 with $K=4$, $U=3$, and $T=0$.
Consider the case where
$
\mathcal U_1=\{1,2,3,4\}
$
and
$
\mathcal U_2=\{2,3,4\}.
$
That is, all users survive the first round, while User~1
drops out in the second round.

In this case, the surviving users in $\mathcal U_2$
must still recover the desired aggregation
$
W_1+W_2+W_3+W_4,
$
even though User~1 does not transmit any second-round
message. Since the input $W_1$ is available only at User~1,
the information about $W_1$ contained in the final decoding
must already be carried by the first-round transmission $X_1$.
Intuitively, if $X_1$ does not contain sufficient information
about $W_1$, then after User~1 drops out in the second round,
the surviving users cannot reconstruct the desired sum.
We next formalize this intuition.
For User~2, the correctness constraint implies
\begin{align}
    0
    =&H\Big(W_1+W_2+W_3+W_4\Big|\notag\\
    &X_1,X_3,X_4,Y_3^{[4]},Y_4^{[4]},W_2,Z_2\Big)\\
    \geq&
    H\Big(W_1+W_2+W_3+W_4\Big|
    X_1,X_3,X_4,\notag\\
    &Y_3^{[4]},Y_4^{[4]},
    W_2,Z_2,W_3,Z_3,W_4,Z_4\Big)\\
    \overset{(\ref{messageX})(\ref{messageY})}{=}&
    H\Big(W_1+W_2+W_3+W_4\Big|\notag\\
    &X_1,W_2,Z_2,W_3,Z_3,W_4,Z_4\Big)\label{covex1t1}\\
    =&
    H\Big(W_1\Big|
    X_1,W_2,Z_2,W_3,Z_3,W_4,Z_4\Big).\label{covex1t2}
\end{align}

Equation~(\ref{covex1t1}) follows because
$
X_3,X_4,Y_3^{[4]},Y_4^{[4]}
$
are deterministic functions of
$
W_3,Z_3,W_4,Z_4.
$
Therefore, after conditioning on
$
W_2,Z_2,W_3,Z_3,W_4,Z_4,
$
the only remaining uncertainty in the desired aggregation
comes from $W_1$. Since User~1 does not participate in the
second round, the information about $W_1$ required for
decoding can only be conveyed through the first-round
message $X_1$.

Now consider the entropy of $W_1$:
\begin{align}
    L
    \overset{(\ref{h2})}{=}&H(W_1)\\
    \overset{(\ref{ind})}{=}&H(W_1|W_2,Z_2,W_3,Z_3,W_4,Z_4)\\
    =&I(W_1;X_1|W_2,Z_2,W_3,Z_3,W_4,Z_4)\notag\\
    &+\underbrace{
    H(W_1|W_2,Z_2,W_3,Z_3,W_4,Z_4,X_1)
    }_{\overset{(\ref{covex1t2})}{=}0}\\
    \leq&
    H(X_1|W_2,Z_2,W_3,Z_3,W_4,Z_4)\\
    \leq&
    H(X_1)\\
    \leq&
    L_X.
\end{align}

Therefore,
\begin{align}
    R_1=\frac{L_X}{L}\geq 1.
\end{align}

This simple example captures the key intuition behind the
general converse proof: if a user may drop out in the second
round while its input must still be included in the desired
aggregation, then the first-round transmission must already
contain sufficient information about that user's input.

\subsection{Example 1: Converse for Second-Round Communication} \label{conexpfr2}

We next illustrate the converse idea for the second-round
communication using Example~1 with
$K=4,U=3$ and $T=0
$.
Unlike the converse for the first-round communication,
the converse for the second-round communication
 relies on the security constraint.
Intuitively, security forces part of the communication
to become useless for decoding, and hence the remaining
second-round messages must carry all information required
for recovering the desired aggregation.

Consider
$
\mathcal U_1=\{1,2,3,4\}
$
and
$
\mathcal U_2=\{1,2,3\}.
$
Since $T=0$, the security constraint requires that each user
learns nothing beyond the desired aggregation.
We focus on User~1.
From the security constraint, User~1 together with the
first-round messages and part of the second-round messages
should not obtain any useful information about the desired
aggregation beyond the target sum itself. Specifically,
\begin{align}
0
\overset{(\ref{security})}{=}&
I\Big(W_1,W_2,W_3,W_4;X_2,X_3,X_4,Y_2^{[4]},Y_3^{[4]},Y_4^{[4]}\notag\\
&\Big|W_1,Z_1,W_1+W_2+W_3+W_4\Big)\label{ex1ry1}\\
\geq&I\Big(W_1+W_2+W_3;X_2,X_3,\notag\\
&\Big|W_1,Z_1,W_1+W_2+W_3+W_4\Big)\label{ex1ry2}\\
=&I\Big(W_1+W_2+W_3;X_2,X_3,\notag\\
&W_1,Z_1,W_1+W_2+W_3+W_4\Big)\notag\\
&-\underbrace{I\Big(W_1+W_2+W_3; W_1,Z_1,W_1+W_2+W_3+W_4\Big)}_{\overset{(\ref{ind})}{=}0}\label{ex1ry3}\\
\geq&I\Big(W_1+W_2+W_3;X_2,X_3,W_1,Z_1\Big).\label{ex1ry4}
\end{align}

Equation~(\ref{ex1ry4}) shows that the first-round
messages do not provide useful information for recovering
the desired aggregation once the side information of
User~1 and the target sum are given. Consequently,
the remaining useful information required for decoding
must come from the second-round messages.

Next, from the correctness constraint,
User~1 must recover
$
W_1+W_2+W_3+W_4
$
from
$
X_2,X_3,Y_2^{[4]},Y_3^{[4]},W_1,Z_1.
$
Therefore,
\begin{align}
L=&H(W_1+W_2+W_3)\\
=&I\Big(W_1+W_2+W_3;X_2,X_3,Y_2^{[3]},Y_3^{[3]},W_1,Z_1\Big)\notag\\
&+\underbrace{H\Big(W_1+W_2+W_3\Big|X_2,X_3,Y_2^{[3]},Y_3^{[3]},W_1,Z_1\Big)
}_{\overset{(\ref{correctness})}{=}0}\\
=&\underbrace{I\Big(W_1+W_2+W_3;X_2,X_3,W_1,Z_1\Big)}_{\overset{(\ref{ex1ry4})}{=}0}\notag\\
&+I\Big(W_1+W_2+W_3;Y_2^{[3]},Y_3^{[3]}\Big|X_2,X_3,W_1,Z_1\Big)\\
\leq&H\Big(Y_2^{[3]},Y_3^{[3]}\Big)\\
\leq&H\Big(Y_2^{[3]}\Big)+H\Big(Y_3^{[3]}\Big)\\
\leq &2L_Y.
\end{align}

Hence,
\begin{align}
R_2=\frac{L_Y}{L}\geq \frac12.
\end{align}

This example illustrates the key converse intuition:
due to the security constraint, part of the communication
becomes useless for decoding, and therefore the remaining
second-round messages must carry all information required
for recovering the desired aggregation.

\subsection{Converse for $R_1 \geq 1$ When $U > T+1$} \label{RX1RY1}

We now establish a lower bound on the first-round communication rate. Intuitively, consider any user $u$ that survives the first round but may drop out in the second round. Since $W_u$ is only available at User $u$, the first-round message $X_u$ must contain sufficient information about $W_u$ so that the desired sum including $W_u$ can still be recovered, which implies that $L_X \geq L$.
We next formalize this intuition.

Consider any $u \in [K]$ and $u'\in [K]\setminus \{u\}$, and set $\mathcal{U}_1 = [K]$, $\mathcal{U}_2 = [K]\backslash\{u\}$. Then from the correctness constraint (\ref{correctness}), we have
\begin{align}
0 \overset{(\ref{correctness})}{=}& H\Bigg(\sum_{k\in[K]} W_k \Bigg| \left\{X_k\right\}_{k\in[K]\setminus\{u'\}}, \left\{Y_k^{[K]}\right\}_{k\in[K]\setminus\{u,u'\}},\notag\\
&W_{u'},Z_{u'} \Bigg) \\
\geq& H\Bigg(\sum_{k\in[K]} W_k \Bigg| \left\{X_k\right\}_{k\in[K]\setminus\{u'\}}, \left\{Y_k^{[K]}\right\}_{k\in[K]\setminus\{u,u'\}},\notag\\
&\{W_k,Z_k\}_{k\in [K]\setminus\{u\}} \Bigg) \\
\overset{(\ref{messageX})(\ref{messageY})}{=}& H\left( W_{u} \Big| X_{u}, \left\{ W_k, Z_k \right\}_{k\in[K]\backslash\{u\}} \right), \label{eq:cc1}
\end{align}
where (\ref{eq:cc1}) follows since $ \left\{X_k \right\}_{k\in[K]\backslash\{u,u'\}}$ and $\{Y_k^{[K]}\}_{k\in[K]\backslash\{u,u'\}}$ are deterministic functions of $\left\{ W_k, Z_k \right\}_{k\in[K]\backslash\{u\}}$ (see (\ref{messageX}), (\ref{messageY})).
We next show that
\begin{eqnarray}
L &\overset{(\ref{h2})}{=}& H(W_{u})\\
&\overset{(\ref{ind})}{=}& H\left( W_{u} \Big|  \left\{ W_k, Z_k \right\}_{k\in[K]\backslash\{u\}} \right) \\
&=& I\left( W_{u} ; X_{u} \Big| \left\{ W_k, Z_k \right\}_{k\in[K]\backslash\{u\}} \right)\notag\\
&&+\underbrace{H\left( W_{u} \Big| X_{u}, \left\{ W_k, Z_k \right\}_{k\in[K]\backslash\{u\}} \right)}_{\overset{(\ref{eq:cc1})}{=}0} \\
&\leq& H\left(X_{u} \Big|  \left\{ W_k, Z_k \right\}_{k\in[K]\backslash\{u\}} \right) \label{eq:cr1} \\ 
&\leq& H\left(X_{u} \right) ~\leq~ L_X, \label{eq:cr2}
\end{eqnarray}
which implies
\begin{eqnarray}
R_1 \overset{(\ref{rate})}{=} \frac{L_X}{L} \geq 1.
\end{eqnarray}

\subsection{Converse for $R_2 \geq \frac{1}{U-T-1}$ When $U > T+1$} \label{R2}

Intuitively, due to the security constraint, all first-round messages and any $T+1$ second-round messages do not contribute useful information for decoding the desired sum (see (\ref{eq:cd})). Hence, all useful information must come from the remaining $(U-T-1)$ second-round messages, so that each of them must carry at least $L/(U-T-1)$ symbols of information on average.
We now formalize this argument.

Consider $\mathcal{U}_1 = [U+1]$, $\mathcal{U}_2 = [U]$, and $\mathcal{T} = [T]$. 
Recall that $U>T+1$. For User $u=T+1$, from the security constraint~(\ref{security}), we have
\begin{align}
0 \overset{(\ref{security})}{=}& I\Bigg(\left\{W_k\right\}_{k\in[K]}; \left\{X_k\right\}_{k\in[K]\setminus\{T+1\}},\notag\\
&\left\{Y_k^{[U+1]}\right\}_{k\in[U+1]\setminus\{T+1\}} \Bigg| \sum_{k\in[U+1]} W_k,\notag\\
& \left\{ W_k, Z_k \right\}_{k\in([T]\cup\{T+1\})} \Bigg) \\
\geq& I\Bigg( \sum_{k\in[U]} W_k; \left\{X_k\right\}_{k\in[U]\setminus\{T+1\}} \Bigg| \notag\\
& \sum_{k\in[U+1]} W_k, \left\{ W_k, Z_k \right\}_{k\in[T+1]} \Bigg) \\
\overset{(\ref{messageY})}{=}& I\Bigg( \sum_{k\in[U]} W_k; \left\{X_k\right\}_{k\in[U]\setminus\{T+1\}}, \left\{Y_k^{[U]}\right\}_{k\in[T]} \Bigg|\notag\\
&\sum_{k\in[U+1]} W_k, \left\{ W_k, Z_k \right\}_{k\in[T+1]} \Bigg)\label{covt1} \\
=&  I\Bigg( \sum_{k\in[U]} W_k; \left\{X_k\right\}_{k\in[U]\setminus\{T+1\}}, \left\{Y_k^{[U]}\right\}_{k\in[T]},\notag\\
&\sum_{k\in[U+1]} W_k, \left\{ W_k, Z_k \right\}_{k\in[T+1]} \Bigg)  \notag\\
&~- \underbrace{ I\Bigg( \sum_{k\in[U]} W_k; \sum_{k\in[U+1]} W_k, \left\{ W_k, Z_k \right\}_{k\in[T+1]} \Bigg) }_{\overset{(\ref{eq:indc})}{=} 0} \label{covt2} \\
\geq& I\Bigg( \sum_{k\in[U]} W_k;\left\{X_k\right\}_{k\in[U]\setminus\{T+1\}}, \left\{Y_k^{[U]}\right\}_{k\in[T]},\notag\\
&W_{T+1},Z_{T+1} \Bigg), \label{eq:cd}
\end{align}
where (\ref{covt1}) holds since the first-round messages $\left\{X_k\right\}_{k\in[U]\setminus\{T+1\}}$ are deterministic functions of $\left\{ W_k, Z_k \right\}_{k\in[T+1]}$ by (\ref{messageX}), (\ref{messageY}), and the second term in (\ref{covt2}) is zero by Lemma~\ref{lemma:indc}.
Next, consider $\mathcal{U}_1 = \mathcal{U}_2 = [U]$. For User $T+1$, we have
\begin{align}
 L =&  H\Bigg(\sum_{k\in[U]} W_k\Bigg)\\
 =&I\Bigg(\sum_{k\in[U]} W_k ;\left\{X_k\right\}_{k\in[U]\setminus\{T+1\}}, \left\{Y_k^{[U]}\right\}_{k\in[U]\setminus\{T+1\}},\notag\\
 &W_{T+1},Z_{T+1}\Bigg) + \underbrace{H\Bigg(\sum_{k\in[U]} W_k \Bigg| \left\{X_k\right\}_{k\in[U]\setminus\{T+1\}},}\notag\\
 &\underbrace{ \left\{Y_k^{[U]}\right\}_{k\in[U]\setminus\{T+1\}},W_{T+1},Z_{T+1} \Bigg)}_{\overset{(\ref{correctness})}{=}0} \label{covtt3}\\
=& I\Bigg(\sum_{k\in[U]} W_k ; \left\{Y_k^{[U]}\right\}_{k\in[U]\backslash[T+1]}  \Bigg| \left\{X_k\right\}_{k\in[U]\setminus\{T+1\}},\notag\\
&\left\{Y_k^{[U]}\right\}_{k\in[T]},W_{T+1},Z_{T+1} \Bigg) + \underbrace{I\Bigg( \sum_{k\in[U]} W_k;}\notag\\
&\underbrace{\left\{X_k\right\}_{k\in[U]\setminus\{T+1\}},\left\{Y_k^{[U]}\right\}_{k\in[T]},W_{T+1},Z_{T+1} \Bigg)}_{\overset{(\ref{eq:cd})}{=}0}\label{covtt4}\\
\leq& H\Bigg( \left\{Y_k^{[U]}\right\}_{k\in[U]\backslash[T+1]}  \Bigg)\notag\\
\leq& \sum_{k\in[U]\backslash[T+1]} H\left( Y_k^{[U]}\right) \leq (U-T-1) L_Y\\
\Rightarrow&~~ R_2 \overset{(\ref{rate})}{=} \frac{L_Y}{L} \geq \frac{1}{U-T-1},
\end{align}
where the second term in (\ref{covtt3}) is zero due to the correctness constraint (\ref{correctness}), and the second term in (\ref{covtt4}) is zero by (\ref{eq:cd}).

\section{Conclusion}\label{sec:conclusion & future directions}

This paper studied DSA with user dropout and $T$-colluding adversaries. We proposed a two-round communication framework that enables correct recovery of the global sum of private inputs while guaranteeing information-theoretic privacy.
We established a feasibility condition, showing that secure aggregation is achievable if and only if $U > T+1$. This result precisely characterizes the interplay between user dropout tolerance and collusion resistance in decentralized settings.
For the feasible regime, we developed an explicit coding scheme based on structured correlated keys and MDS codes. The scheme achieves optimal communication rates, namely $R_1 = 1$ for the first round and $R_2 = \frac{1}{U-T-1}$ for the second round, and thus fully characterizes the rate region under the proposed model.

An interesting future direction is to extend the proposed framework to  arbitrary or sparse networks. With limited neighborhood connectivity, recovery becomes harder because correlated keys must cancel through local, topology-constrained interactions rather than global coordination. Characterizing how network topology affects feasibility, dropout tolerance, collusion resistance, and communication rates remains open.

\bibliographystyle{IEEEtran}
\bibliography{references_secagg.bib}
\end{document}

%% file: macros.tex
\long\def\comment#1{}

\makeatletter
\newcommand{\thickhline}{%
    \noalign {\ifnum 0=`}\fi \hrule height 1pt
    \futurelet \reserved@a \@xhline
}
\newcolumntype{"}{@{\hskip\tabcolsep\vrule width 1pt\hskip\tabcolsep}}
\makeatother

\makeatletter
\newcommand{\subalign}[1]{
  \vcenter{%
    \Let@ \restore@math@cr \default@tag
    \baselineskip\fontdimen10 \scriptfont\tw@
    \advance\baselineskip\fontdimen12 \scriptfont\tw@
    \lineskip\thr@@\fontdimen8 \scriptfont\thr@@
    \lineskiplimit\lineskip
    \ialign{\hfil$\m@th\scriptstyle##$&$\m@th\scriptstyle{}##$\crcr
      #1\crcr
    }%
  }
}
\makeatother

\newtheorem{theorem}{Theorem}
\newtheorem{lemma}{Lemma}
\newtheorem{remark}{Remark}

\let\trm\textrm

\let\tit\textit

\newcommand{\eg}{e.g.\xspace}

\newcommand{\Thm}{Theorem\xspace}
\newcommand{\schm}{scheme\xspace}

\newcommand{\agg}{aggregation\xspace}

\newcommand{\indiv}{individual\xspace}

\newcommand{\comm}{communication\xspace}

\newcommand{\achvb}{achievable\xspace}

\newfont{\bbb}{msbm10 scaled 700}

\newfont{\bb}{msbm10 scaled 1100}

\newcommand{\Rc}{{\cal R}}

\newcommand{\be}{\begin{equation}}
\newcommand{\ee}{\end{equation}}
\newcommand{\bea}{\begin{eqnarray}}
\newcommand{\eea}{\end{eqnarray}}



%% file: author_journal.tex
\author{
Zhou Li,~\IEEEmembership{Member,~IEEE},
Xiang~Zhang,~\IEEEmembership{Member,~IEEE},
Yizhou Zhao,~\IEEEmembership{Member,~IEEE},
Han Yu,~\IEEEmembership{Member,~IEEE},
and Giuseppe Caire,~\IEEEmembership{Fellow,~IEEE}

\thanks{Z. Li is with the Guangxi Key Laboratory of Multimedia Communications and Network Technology, 
Guangxi University, Nanning 530004, China (e-mail: lizhou@gxu.edu.cn).}

\thanks{X. Zhang, H. Yu, and G. Caire are with the Department of Electrical Engineering and Computer Science, Technical University of Berlin, 10623 Berlin, Germany (e-mail: \{xiang.zhang, caire\}@tu-berlin.de, han.yu.1@campus.tu-berlin.de).
}

\thanks{Y. Zhao is with the College of Electronic and Information Engineering, Southwest University, Chongqing, China (e-mail: onezhou@swu.edu.cn).}

}

